%% file: main.tex
\documentclass[letterpaper, 10 pt, conference]{ieeeconf}  % Comment this line out if you need a4paper

\IEEEoverridecommandlockouts                              % This command is only needed if 
\usepackage{amsmath} % assumes amsmath package installed
\usepackage{amssymb}  % assumes amsmath package installed
\usepackage{xcolor}
\usepackage{cite}
\usepackage{caption}
\usepackage{subcaption}

\usepackage{amssymb}
\usepackage{dsfont}
\usepackage{cite}
\usepackage{mathtools}
\usepackage{threeparttable}
\usepackage{colortbl}
\usepackage{siunitx}
\usepackage{multirow}
\usepackage{graphicx}
\usepackage{epstopdf}
\usepackage{times}     % 设置 Times New Roman 字体
\usepackage{balance}
\usepackage{multirow}
\usepackage{booktabs}
\usepackage{multicol}
\usepackage{epstopdf}

\DeclareSIUnit[]{\pu}{p.u.}
\DeclareSIUnit[]{\VA}{VA}

\DeclareSymbolFont{bbold}{U}{bbold}{m}{n}
\DeclareSymbolFontAlphabet{\mathbbold}{bbold}

\DeclarePairedDelimiterX\Set[2]{\lbrace}{\rbrace}%
{ #1 \,\delimsize| \,\mathopen{} #2 }

\let\labelindent\relax
\usepackage{enumitem}
\setlist[itemize]{leftmargin=*}
\usepackage{tikz}

\allowdisplaybreaks
\newtheorem{thm}{Theorem}

\newtheorem{rem}{Remark}
\newtheorem{lem}{Lemma}
\newtheorem{ass}{Assumption}

\title{\LARGE \bf Learning-Augmented Primal-Dual Control Design for Secondary Frequency Regulation}

\author{Yixuan Yu, Rajni K. Bansal, Yan Jiang, and Pengcheng You% <-this % stops a space
\thanks{P. You is the co-corresponding author.}
\thanks{This work was supported in part by the National Natural Science Foundation of China under Grant 72431001, in part by CUHKSZ University Development Fund, and in part by the Brij Disa Centre for Data Science and Artificial Intelligence (CDSA) at the Indian Institute of Management Ahmedabad (IIMA).}
\thanks{Y. Yu and P. You are with the Department of Control Science and Systems Engineering, Peking University, Beijing, China.
        {\tt\small yuyixuan834@stu.pku.edu.cn, pcyou@pku.edu.cn}}%
\thanks{R. K. Bansal is with the Operations and Decision Sciences Area, Indian Institute of Management Ahmedabad, Gujarat, India. {\tt\small rajnikantb@iima.ac.in}}%  
\thanks{Y. Jiang is with the School of Science and Engineering, The Chinese University of Hong Kong, Shenzhen, China.
        {\tt\small yjiang@cuhk.edu.cn}}%      
}

\begin{document}

\maketitle

\thispagestyle{empty}
\pagestyle{empty}

%%%%%%%%%%%%%%%%%%%%%%%%%%%%%%%%%%%%%%%%%%%%%%%%%%%%%%%%%%%%%%%%%%%%%%%%%%%%%%%%
\begin{abstract}
\input{Abstract}
\end{abstract}

%%%%%%%%%%%%%%%%%%%%%%%%%%%%%%%%%%%%%%%%%%%%%%%%%%%%%%%%%%%%%%%%%%%%%%%%%%%%%%%%
\section{Introduction}\label{sec: Introduction}
\input{Introduction}

\section{Problem Statement}\label{sec: Problem Statement}
\input{Problem_Statement}

\section{Controller Design and Analysis}\label{sec: General Form of Controller}
\input{General_Form_of_Controller}

\section{Learning for Transient Improvement}\label{sec: Neural Network}
\input{Neural_Network}

\section{Numerical Results}\label{sec: Numerical Results}
\input{Numerical_Results}

\section{Conclusion}\label{sec: Conlusion}
\input{Conclusion}

%\addtolength{\textheight}{-12cm}   % This command serves to balance the column lengths
                % the textheight of the last page by a suitable amount.
                                  % This command does not take effect until the next page
                                  % so it should come on the page before the last. Make
                                  % sure that you do not shorten the textheight too much.

%%%%%%%%%%%%%%%%%%%%%%%%%%%%%%%%%%%%%%%%%%%%%%%%%%%%%%%%%%%%%%%%%%%%%%%%%%%%%%%%

%%%%%%%%%%%%%%%%%%%%%%%%%%%%%%%%%%%%%%%%%%%%%%%%%%%%%%%%%%%%%%%%%%%%%%%%%%%%%%%%

%%%%%%%%%%%%%%%%%%%%%%%%%%%%%%%%%%%%%%%%%%%%%%%%%%%%%%%%%%%%%%%%%%%%%%%%%%%%%%%%

%%%%%%%%%%%%%%%%%%%%%%%%%%%%%%%%%%%%%%%%%%%%%%%%%%%%%%%%%%%%%%%%%%%%%%%%%%%%%%%%

% \section*{ACKNOWLEDGMENT}

% The preferred spelling of the word ÒacknowledgmentÓ in America is without an ÒeÓ after the ÒgÓ. Avoid the stilted expression, ÒOne of us (R. B. G.) thanks . . .Ó  Instead, try ÒR. B. G. thanksÓ. Put sponsor acknowledgments in the unnumbered footnote on the first page.

\bibliographystyle{IEEEtran}
\bibliography{root}          % on the last page of the document manually. It shortens

\section*{APPENDIX}\label{sec:appendix}
\input{Appendix}
                  
\end{document}

%% file: Abstract.tex
% PCY try
Frequency stability is fundamental to the secure operation of power systems. With growing uncertainty and volatility introduced by renewable generation, secondary frequency regulation - responsible for restoring the nominal system frequency - must now deliver enhanced performance not only in the steady state but also during transients. This paper presents a systematic framework to embed learning in the design of a primal–dual controller that provides provable (potentially exponential) stability and steady-state optimality, while simultaneously improving key transient metrics, including frequency nadir and control effort, in a data-driven manner. In particular, we employ the primal-dual dynamics of an optimization problem that encodes steady-state objectives to realize secondary frequency control with asymptotic stability guarantee. To augment transient performance of the controller via learning, a change of variables on control inputs, which will be deployed by neural networks, is proposed such that under mild conditions, stability and steady-state optimality are preserved. It further allows us to define a learning goal that accounts for the exponential convergence rate, frequency nadir and accumulated control effort, and use sample trajectories to enhance these metrics. Simulation results validate the theories and demonstrate superior transient performance of the learning-augmented primal-dual controller.

%% file: Introduction.tex
%关于这篇工作和DAI的差别，通过简要的话说明能够包含DAI，体现有显著不同，它是special case
The stability of power systems requires maintaining the frequency close to its nominal value (50 Hz or 60 Hz), which is essential for the safe and reliable operation of grid equipment. From a time-scale perspective, frequency control is traditionally divided into three layers~\cite{machowski1997power}: primary control for immediate stabilization, secondary control to restore the frequency to its nominal value, and tertiary control for economic dispatch to optimize steady-state generation.

%Recent research has sought to achieve secondary frequency control and economic optimality at the same time. A common approach is to cast the problem as an optimization task, where the objective function models the economic cost and the constraints capture the steady-state conditions of frequency regulation. In this framework, both the system dynamics and controller updates can be interpreted as the primal-dual dynamics of the optimization problem, thereby aligning the control objective with the economic objective.
Recent research has sought to achieve secondary frequency control and economic optimality at the same time, which can be broadly categorized into two classes of approaches. The first class leverages consensus-based algorithms to enhance the efficiency of classical PI controllers~\cite{andreasson2014distributed, dorfler2015breaking, shafiee2013distributed, Zhao2015acc}. The second approach models the objectives of frequency regulation and economic optimality jointly as an optimization problem, where the closed-loop system is interpreted through primal–dual dynamics~\cite{Zhao:2014bp, li2016tcns, cai2017distributed, mallada2017optimal, Wang2019tsg}.
Particularly, the primal-dual framework has attracted significant attention, as it provides a unifying perspective and bridges the control objective of frequency regulation with the economic objective of optimal power system operation. However, they are all focusing on the optimality of the steady state, ignoring the transient behavior of power systems.
%For example, \cite{Zhao:2014bp} cast load-side primary frequency control as an optimal load control problem, showing that swing dynamics with load control implement a distributed primal-dual algorithm. Similarly, \cite{li2016tcns} reformulate automatic generation control (AGC) with economic optimality and interpret its mechanism as the primal-dual dynamics of the corresponding optimization problem. In \cite{cai2017distributed, mallada2017optimal, Wang2019tsg}, the primal-dual framework is further leveraged to embed operational constraints into controller design. 

In practice, transient performance is also critical for the safety and efficiency of power systems. %Although various controllers can ensure frequency regulation and steady-state economic optimality, their transient behaviors may differ significantly~\cite{guha2022performance, srikanth2023improvement}. 
Some works~\cite{badal2021robust, das2023frequency} improve transient performance using model-based feedback control methods, such as LQR and MPC. These methods typically rely on linear feedback since it is challenging to incorporate nonlinear control laws into the optimal control problem. Nevertheless, nonlinear designs become important when state deviations are large, as they can capture richer system dynamics more accurately than linear approximations. This motivates the use of learning-based approaches, which offer strong capability in parameter optimization.

Another line of works has employed neural network (NN) to achieve better transient performance in frequency regulation tasks. Some of them~\cite{Cui2023tps, sun2023ESS} improve transient performance by parameterizing the controller with a monotone neural network and training its parameters. However, these approaches do not account for steady-state economic optimality. To address this limitation, a nonlinear extension of the linear distributed averaging-based integral (DAI) controller in~\cite{Zhao2015acc} is proposed in \cite{jiang2022ojcsys}, where the nonlinear component is parametrized by a monotone neural network. Although this design ensures steady-state optimality and enhances transient performance, it remains restricted in terms of the optimality objectives it can accommodate, such as frequency regulation task with constrained optimality. 
Taken together, these works consistently rely on monotone structures in their controllers. However, the use of monotonicity is typically assumed rather than derived, lacking a clear rationale or systematic justification of its underlying mechanism. Additionally, they employ similar transient metrics, yet these metrics offer little theoretical guarantee about the destination of convergence in learning.
This naturally raises a key question: \textit{how can we systematically design a controller that achieves secondary frequency regulation and steady-state economic optimality, while allowing learning to augment transient performance under theoretical guarantee?}
%More collectively, these works consistently adopt monotone structures in their controllers. However, monotonicity is usually assumed in these controllers, without fundamental rationale or systematic explanation of its underlying mechanism.

%In fact, many works \cite{Cui2023tps, jiang2022ojcsys, sun2023ESS, feng2023bridging} have employed neural network (NN) as a tool to parameterize parts of the controllers. Among them, \cite{Cui2023tps} and \cite{sun2023ESS} enforce the control input to be monotone with respect to frequency in frequency regulation tasks, improving transient performance through learning monotone neural networks. But they fail to take steady-state economic optimality into consideration. \cite{jiang2022ojcsys} generalizes the linear distributed averaging-based integral (DAI) controller~\cite{Zhao2015acc}. This generalized DAI achieves secondary frequency regulation and steady-state economic optimality through neighborhood communication, but also enhances transient performance by learning a monotone nonlinear component. These works consistently adopt monotone structures in their controllers. However, monotonicity is usually assumed in these controllers, without a clear structural explanation of where this monotonicity originates or why it works.
%without fundamental rationale or systematic explanation of its underlying mechanism

To address this, we extend the primal-dual framework to design a controller with nonlinear, learnable component. Specifically, we model the optimal steady-state objectives jointly as an optimization problem and reinterpret the controller as part of its associated primal-dual dynamics. This interpretation provides a principled way to incorporate nonlinear control law without compromising asymptotic stability or steady-state economic optimality. Building upon this, the nonlinear learnable component is parameterized by a monotone neural network, enabling the enhancement of transient performance. To further ensure the effectiveness of learning, we introduce metrics for transient performance and, in turn, provide theoretical potential of exponential convergence.

Our contributions are summarized as follows.
\begin{enumerate}
    \item \textbf{Unified framework for steady state and transient performance:} We develop a unified framework for controller design that simultaneously guarantees optimal steady-state objectives and augments transient performance. Specifically, we reinterpret the controller as a learning embedded primal–dual dynamics. The asymptotic stability and economic optimality are ensured by the underlying primal-dual structure, while transient enhancement is achieved through learning embedded in the framework.

    \item \textbf{Learning flexibility via nonlinear change of variables:} The learnable component in the controller originates from a change of variables from an optimization perspective, systematically embedding flexibility for learning into the primal-dual based controller design. Importantly, strict monotonicity of the change of variables ensures the preservation of steady-state stability and optimality.

    \item \textbf{Convergence rate metric for transient improvement:} We propose a metric to evaluate the convergence rate of the controller and guide the learning process. We further prove, theoretically, that the closed-loop system achieves potentially exponential stability after learning, based on this metric.
\end{enumerate}

%% file: Problem_Statement.tex
\subsection{Power System Model}
We consider a $n$-bus power network with the topology characterized by a directed and connected graph $\left(\mathcal{V},\mathcal{E} \right)$, where the buses indexed by $i \in \mathcal{V} :=\{1,\dots, n\} $ are connected through the transmission lines indexed by ordered pairs $(i,j) \in \mathcal{E} \subset \mathcal{V}\times \mathcal{V}$.  %\Set*{\{i,j\}}{i,j\in\mathcal{V},i\not=j}$. 
Without loss of generality, we assume there is only one aggregate controllable generator at each bus subject to a linear approximation of swing dynamics \cite{you2018stabilization}. 
%, whose dynamics is described with a second-order differential equation. 
More precisely, given the power disturbance $p_{i}$ 
at bus $i\in\mathcal{V}$, represented by a step change scalar, the dynamics of the phase angle $\theta_i$ and the frequency $\omega_i$ from the nominal frequency are given by 
\begin{subequations}\label{eq: sys-dyn-theta}
\begin{align}
    \dot{\theta}_i &=\omega_i\,,\label{frequency dynamics}\\
    M_i \dot{\omega}_i&= u_{i} - p_{i}- D_i \omega_i-\sum_{j=1}^n B_{ij}\left(\theta_i-\theta_j\right)\,,  \label{swing equation}
\end{align}
\end{subequations}
where $M_i>0$ is the generator inertia constant, $D_i>0$ summarizes the generator damping and frequency-dependent load, and $B_{ij}>0$ characterizes the sensitivity of the power flow on line $(i,j)\in \mathcal{E}$ to the phase angle differences between its two end buses.
%, if $(i,j) \in\mathcal{E}$ or $(j,i)\in \mathcal{E}$. 
$u_{i}$ is a controllable power injection, by generator, to be designed for frequency control. Note that the network model \eqref{eq: sys-dyn-theta} implicitly assumes that the variables $\theta_i$, $\omega_i$, $u_i$ are deviations from their nominal values.

For convenience in analysis, we define $\tilde\theta_{ij} := \theta_i - \theta_j$ as phase angle difference on line $(i,j) \in \mathcal{E}$ and rewrite the swing dynamics in a compact form:
\begin{subequations}\label{eq: swing dynamics in vector form}
    \begin{align}
        \dot{\tilde\theta}&=C^T\omega\,,\label{eq: frequency dynamics in vector form}\\
        M\dot{\omega}&=u-p-D\omega-CB\tilde\theta\,, \label{eq: swing equation in vector form}
    \end{align}
\end{subequations}
where $\tilde\theta:=(\tilde\theta_{ij}, (i,j)\in\mathcal{E})$, $\omega:=\left(\omega_i, i \in \mathcal{V} \right)$, $u:=(u_i, i\in\mathcal{V})$, and $p:=(p_i, i\in\mathcal{V})$. $M:=\operatorname{diag}(M_i, i\in\mathcal{V})$, $D:=\operatorname{diag}(D_i, i\in \mathcal{V})$, and $B:=\operatorname{diag}(B_{ij}, (i,j)\in\mathcal{E})$ are diagonal positive definite matrices. $C \in \mathbb{R}^{|\mathcal{V}| \times |\mathcal{E}|}$ is the incidence matrix of the directed graph $(\mathcal{V}, \mathcal{E})$.

\subsection{Optimal Steady-State Objectives}
The basic goal of secondary frequency control for the power system is to not only stabilize the frequencies, i.e., $\dot {\omega} = 0$, but also restore them to the nominal value such that $\omega=0$. Under these steady-state conditions, \eqref{eq: swing equation in vector form} reduces to a characterization of the nodal power balance over the network:
\begin{equation}\label{eq:nodal_balance}
    u - p - CB\tilde\theta = 0 \,.
\end{equation}

Beyond this, we further require that at the steady state, the control efforts have to be optimally allocated according to a given cost function $F(u) := \sum_{i=1}^n F_i(u_i)$ for the power injection at each bus $i$. We assume that $F(u):\mathbb{R}^n \to \mathbb{R}$ is strictly convex and second-order continuously differentiable with its gradient satisfying $\nabla F(u) = 0$ at $u=0$.\footnote{The condition $\nabla F(0)=0$ is natural from a physical standpoint, since it implies that zero control input corresponds to the minimum operating cost.}

Therefore, given $p$, the optimal steady-state problem that minimizes the aggregate control cost to achieve secondary frequency control is given by
\begin{subequations}\label{eq: opt-ss}
\begin{align}
    \underset{\tilde\theta, \omega, u}\min\quad &F(u) =\sum_{i=1}^nF_i(u_i)\\
	{\operatorname*{s.t.}}\quad & u-p-D\omega -CB\tilde\theta = 0 \label{eq: omega dot=0}\\
    & u-p-CB\tilde\theta = 0\,,\label{eq: nodal power balance in optimiation problem}
\end{align}
\end{subequations}
where \eqref{eq: omega dot=0} and \eqref{eq: nodal power balance in optimiation problem} can jointly enforce $\dot{\omega}=0$ and $\omega=0$ at optimal.

\subsection{Transient Performance Objectives}
%除了secondary frequency control和economic dispatch，我们最重要的目的是能够对一些指标进行优化：nadir + cost + speed（前两者在理论部分无法说明，speed在理论部分会需要通过preconditioning进行解释）
%直接介绍三种transient metrics，不用强调哪些理论保证哪些没有
Although the optimal steady-state objectives guarantee desirable long-run behavior, it lacks a mechanism to shape transient dynamics. In practice, the quality of power system operation strongly depends on transient behaviors, such as the speed of recovery (convergence rate), the maximum frequency deviation (frequency nadir), and the accumulated control effort required throughout the process.

%As mentioned in Section~\ref{sec:Introduction}, nonlinear control laws are challenging to formalize, despite their greater potential to enhance system stability.Then a natural approach is to leverage learning as a powerful tool for parameter optimization \cite{jiang2022ojcsys}
To employ learning for transient performance improvement, the key challenge is integrating sufficient flexibility into the controller design to apply learning effectively, while still achieving optimal steady-state objectives. Furthermore, we aim to design appropriate transient performance metrics that can guide the learning process, ensuring the resulting controller achieves favorable performance. 
% according to the metrics of interest.

%% file: General_Form_of_Controller.tex
%总体描述这一节的内容：首先第一部分是针对optimal secondary frequency control来回顾primal-dual based controller；第二部分是基于优化的视角从preconditioning的角度引入变量替换实现控制器速度的提升，给学习提供自由度
%可能存在一个gap：理论上解释的时候precondition和变量替换是为了提升速度，但是最后学习的时候指标是nadir+speed+cost可能解释一下，结合前一个section。

%先总写controller的形式（包括learning的部分的形式），但需要简短但总结性地给出那一部分是model-based，哪一部分是通过learning实现的
%介绍完总体形式之后开始介绍设计controller的原因和步骤。第一部分primal-dual和precondition一起说不用分开，因为都是从优化的角度来设计出这样一个控制器。
%第二部分是稳定性的分析，首先通过Lyapunov函数证明渐进稳定，然后说明如果参数理想，loss function能够实现指数稳定

%为了实现上述两个目标，我们给出如下的控制器，其中新引入了变量s, lambda, phi。
%从优化的视角，这样的变量更新方式能够保证闭环系统\eqref{eq: swing dynamics in vector form}\eqref{eq: generalized controller}的平衡点等价于优化问题\eqref{eq: opt-ss}的最优点，从而从model-base的角度保证最优稳态的实现；另一方面，可选择的s与u之间的关系为学习提供了空间，因此通过将控制器f参数化为具有特殊结构的神经网络，我们能实现对关心的暂态目标的优化。最后，我们将证明这样的变量更新方式结合神经网络的特殊结构将共同保证系统的渐进稳定。
In this section, we build upon the primal-dual framework for secondary frequency regulation to design a controller that not only achieves the optimal steady-state objectives but also provides sufficient flexibility for learning. We start by presenting the proposed controller. Then we analyze the controller from a primal-dual perspective, revealing the underlying structure that enables flexibility for learning without compromising stability and steady-state economic optimality. 

To achieve the optimal steady-state objectives and better transient performance at the same time, we propose the following controller\footnote{We avoid using $p$ directly in practice, as disturbances are unmeasurable. Following~\cite{Wang2019tsg}, it can be replaced with $-M\dot{\omega}+f(s)-D\omega-CB\tilde\phi$.}:
\begin{equation}\label{eq: u=f(s)}
    u = f(s)
\end{equation}
with
\begin{subequations}\label{eq: learnable controller}
    \begin{align}
 \dot{s}&=-\big[\nabla F(f(s))+\omega+\lambda\big] \,,\label{eq: s-dot}\\ 
  \dot{\lambda}&=\Gamma^\lambda\big[f(s)-p-CB\tilde\phi\big]\,,\\
 \dot{\tilde\phi}&=\Gamma^{\tilde\phi}\big[BC^T\lambda\big]\,,
    \end{align}
\end{subequations}
where $s:=(s_i, i\in\mathcal{V}), \lambda:=(\lambda_i, i\in\mathcal{V}), \tilde\phi:=(\tilde\phi_{ij}, (i,j)\in\mathcal{V})$ are newly introduced internal variables and $\Gamma^\lambda, \Gamma^{\tilde\phi}  \succ0$ denote diagonal constant control gains. $f(s):=(f_i(s_i), i \in \mathcal{V})$ is a nonlinear function.

The update rules \eqref{eq: learnable controller} of these variables are demonstrated to enforce the optimal steady-state objectives based on the system dynamics, i.e., in a model-based manner. At the same time, the nonlinear $u=f(s)$ is left unspecified to be learned, aiming at enhancing the transient performance.

%We further make the following assumption on the cost function, which the proposed controller is designed to optimize.
%\begin{ass}[Cost function $F$]\label{ass: F}
%    The cost function $F: \mathbb{R}^n \to \mathbb{R}$ is strictly convex and second-order continuously differentiable with its gradient satisfying $\nabla F(u) = \mathbbold{0}_n$ at $u=\mathbbold{0}_n$.
%\end{ass}

%Assumption~\ref{ass: F} broadens the class of cost functions that can be incorporated into the secondary frequency control framework, extending beyond the commonly used quadratic forms. This generalization allows the controller to accommodate a wider range of economic objectives in power system operation.

\subsection{Primal-Dual Interpretation}
To systematically analyze the proposed learnable controller and reveal the underlying mechanism that provides flexibility for learning, we interpret it from a primal-dual perspective. In particular, we construct the following optimization problem, whose primal-dual dynamics aligns with the closed-loop system \eqref{eq: swing dynamics in vector form}, \eqref{eq: u=f(s)}, and \eqref{eq: learnable controller}.
\begin{subequations}\label{eq: general optimization problem with precondition}
\begin{align}
    \min_{\tilde\theta, \omega, s,\tilde\phi} &\quad F(f(s))+\frac{1}{2}\omega^TD\omega \label{eq: general objective function for s}\\
    {\operatorname*{s.t.}} &\quad f(s)-p-D\omega-CB\tilde\theta=0 &&:\nu\label{eq: swing dynamics at equilibium for s}\\
    &\quad f(s)-p-CB\tilde\phi=0 &&:\lambda\label{eq: power balance at equilibrium for s}
\end{align}
\end{subequations}
where $\nu \in \mathbb{R}^n$ and $\lambda \in \mathbb{R}^n$ are Lagrange multipliers of constraints \eqref{eq: swing dynamics at equilibium for s} and \eqref{eq: power balance at equilibrium for s}, respectively. Accordingly, the Lagrangian of \eqref{eq: general optimization problem with precondition} is 
\begin{equation}\label{eq: Lagrangian for general optimization problem with precondition}
\begin{aligned}
    L(\tilde\theta, \omega, s, \nu, \lambda, \tilde\phi):=&F(f(s))+\frac{1}{2}\omega^TD\omega\\
    &+\nu^T(f(s)-p-D\omega-CB\tilde\theta)\\
    &+\lambda^T(f(s)-p-CB\tilde\phi)\,.
\end{aligned}
\end{equation}

Similarly to \cite{you2021tac}, the swing dynamics~\eqref{eq: swing dynamics in vector form} can be expressed as 
\begin{subequations}
    \begin{align}
        \dot{\tilde\theta} &= -\Gamma^{\tilde\theta} \nabla_{\tilde\theta} L(\tilde\theta, \omega, s, \nu, \lambda, \tilde\phi)\,,\\
        \omega &= \arg \min_{\omega}L(\tilde\theta, \omega, s, \nu, \lambda, \tilde\phi)\,, \label{eq: omega=nu}\\
        \dot{\nu} &= \Gamma^\nu\nabla_\nu L(\tilde\theta, \omega, s, \nu, \lambda, \tilde\phi)\,,
    \end{align}
\end{subequations}
with $\Gamma^{\tilde\theta} = B^{-1}$ and $\Gamma^\nu = M^{-1}$. Note that \eqref{eq: omega=nu} enforces $\omega \equiv \nu$ along the trajectory. Hence, we can use $\omega \leftrightarrow \nu$ interchangeably and define the following reduced Lagrangian:
\begin{equation}
    \begin{aligned}
        \tilde{L}(\tilde\theta, \omega, s, \lambda, \tilde\phi):=&F(f(s))+\frac{1}{2}\omega^TD\omega\\
    &+\omega^T(f(s)-p-D\omega-CB\tilde\theta)\\
    &+\lambda^T(f(s)-p-CB\tilde\phi)\,.
    \end{aligned}
\end{equation}
Then, the closed-loop system can be equivalently rewritten as
\begin{subequations}\label{eq: closed-loop system in primal-dual}
    \begin{align}
        \dot{\tilde\theta} &= -\Gamma^{\tilde\theta} \nabla_{\tilde\theta} \tilde{L}(\tilde\theta, \omega, s, \lambda, \tilde\phi)\,,\\
        \dot{\omega}&=\Gamma^\omega\nabla_\omega \tilde{L}(\tilde\theta, \omega, s, \lambda, \tilde\phi)\,,\\
        \dot{s} &= -(J_f(s))^{-1}\nabla_s \tilde{L}(\tilde\theta, \omega, s, \lambda, \tilde\phi) \,,\label{eq: s-dot in primal-dual}\\
        \dot{\lambda} &=\Gamma^\lambda\nabla_\lambda \tilde{L}(\tilde\theta, \omega, s, \lambda, \tilde\phi)\,,\\
        \dot{\tilde\phi} &= -\Gamma^{\tilde\phi}\nabla_{\tilde\phi}\tilde{L}(\tilde\theta, \omega, s, \lambda, \tilde\phi)\,,
    \end{align}
\end{subequations}
with $\Gamma^\omega=\Gamma^\nu, \Gamma^\lambda, \Gamma^{\tilde\phi} \succ 0$ being diagonal step-size matrices. $J_f(s):=\operatorname*{diag}(f_i'(s_i), i \in \mathcal{V})$ is the Jacobian of $f(\cdot)$ to $s$.

\begin{rem}[Precondition]\label{rem: precondition}
    Interpreting the controller within the primal-dual framework reveals a key role of nonlinear structure $u = f(s)$. From an optimization perspective, the change of variables $u=f(s)$ acts as a preconditioning mechanism: This technique improves the conditioning of a optimization problem and accelerates convergence to the optimal solution by modifying the gradient’s magnitude and direction~\cite{meijerink1977iterative}. In our formulation, the gradient flow of problem~\eqref{eq: general optimization problem with precondition} is rescaled, as reflected in the additional step-size matrix $(J_f(s))^{-1}$ in \eqref{eq: s-dot in primal-dual}.  
\end{rem}

We also highlight the role of introducing $\tilde\phi$ as virtual phase angle differences in \eqref{eq: power balance at equilibrium for s} instead of directly using $\tilde\theta$. Basically, this substitution enables the swing dynamics \eqref{eq: swing dynamics in vector form} to be embedded into the primal-dual dynamics of problem~\eqref{eq: general optimization problem with precondition}. Otherwise, the gradient of $\tilde{L}(\tilde\theta, \omega, s, \lambda, \tilde\phi)$ with respect to $\tilde\theta$ would not align with the frequency dynamics~\eqref{eq: frequency dynamics in vector form} that are determined by the physics.

%Consequently, the closed-loop system composed of \eqref{eq: swing dynamics in vector form}, \eqref{eq: u=f(s)}, and \eqref{eq: learnable controller} is reinterpreted as the primal-dual gradient flow of $\tilde{L}(\tilde\theta, \omega, s, \lambda, \tilde\phi)$.
\subsection{Stability Analysis}
Having reinterpreted the closed-loop system \eqref{eq: swing dynamics in vector form}, \eqref{eq: u=f(s)}, and \eqref{eq: learnable controller} as the primal-dual gradient flow of $\tilde{L}(\tilde\theta, \omega, s, \lambda, \tilde\phi)$, now we proceed to characterize the equilibrium of primal-dual dynamics, i.e., the equilibrium of the closed-loop system. Note that the designed optimization problem \eqref{eq: general optimization problem with precondition} is inherently non-convex, since the function $f(s)$ in objective function is non-convex and the equality constraints \eqref{eq: swing dynamics at equilibium for s}, \eqref{eq: power balance at equilibrium for s} are non-affine. Therefore, it is nontrivial to ensure that the primal-dual dynamics converge to a primal-dual optimal solution to \eqref{eq: Lagrangian for general optimization problem with precondition}. To this end, we first make an assumption for the nonlinear change of variables $u=f(s)$ as follows:

\begin{ass}[Change of variables $f$]\label{ass: f} $\forall i\in \mathcal{V}$, $f_i: \mathbb{R} \to \mathbb{R}$ is a strictly increasing and Lipschitz continuous function passing through the origin.
\end{ass}

With Assumption~\ref{ass: f} and the strict convexity of $F$, the composition $\nabla F(f(s))$ is also a strictly monotone function passing through the origin. This structure generalizes the principle of linear negative feedback, which underlies many classical controllers, including droop-like laws in power systems. With this, the following lemma shows that problem \eqref{eq: general optimization problem with precondition} has a unique global optimal solution. Then Lemma~\ref{lem: equivalence of two optimization problem} proves that problem~\eqref{eq: general optimization problem with precondition} is equivalent to \eqref{eq: opt-ss}. Finally, Theorem~\ref{thm: closed-loop equilibrium} establishes the equivalence between closed-loop equilibrium and the optimal solution to \eqref{eq: general optimization problem with precondition}, thereby ensuring the achievement of optimal steady-state objectives due to Lemma~\ref{lem: equivalence of two optimization problem}.
\begin{lem}[Uniqueness]\label{lem: unique global optimal}
    Suppose Assumption~\ref{ass: f} holds. The optimization problem \eqref{eq: general optimization problem with precondition} admits a unique global optimal solution $(\tilde\theta^*, \omega^*, s^*, \tilde\phi^*)$. Furthermore, there exists a unique Lagrange multiplier $(\nu^*, \lambda^*)$ such that $(\tilde\theta^*, \omega^*, s^*, \tilde\phi^*, \nu^*, \lambda^*)$ is the only solution satisfying the KKT conditions of problem~\eqref{eq: general optimization problem with precondition}.
\end{lem}

 Lemma~\ref{lem: unique global optimal} ensures that $(\tilde\theta^*, \omega^*, s^*, \tilde\phi^*, \nu^*, \lambda^*)$ is the unique primal-dual optimal solution to problem~\eqref{eq: general optimization problem with precondition}, and that no other feasible point $(\hat{\tilde\theta}, \hat{\omega}, \hat{s}, \hat{\tilde\phi})\neq (\tilde\theta^*, \omega^*, s^*, \tilde\phi^*)$, together with any multipliers $(\nu, \lambda)$, can satisfy the KKT conditions, even if problem~\eqref{eq: general optimization problem with precondition} is nonconvex.
 
\begin{rem}[Hidden convexity]
With Assumption~\ref{ass: f}, the inverse change of variables $s=f^{-1}(u)$ effectively transforms \eqref{eq: general optimization problem with precondition} into a convex optimization problem, as shown in the proof of Lemma~\ref{lem: unique global optimal}. This highlights that problem~\eqref{eq: general optimization problem with precondition} is not arbitrarily constructed, but rather embeds a hidden convex structure~\cite[Chapter 4.2.5]{Boyd2004convex} - can be converted to a convex problem with variable bijection. 
%From an optimization viewpoint, the change of variables $u=f(s)$ can reveal this hidden convexity.
\end{rem}

Upon Lemma~\ref{lem: unique global optimal}, we now turn to investigate the relationship between problem~\eqref{eq: general optimization problem with precondition} and the optimal steady-state problem~\eqref{eq: opt-ss}. 
%Our ultimate goal is to show that the closed-loop equilibrium indeed achieves the optimal steady-state objective. To this end, we next establish the equivalence between the optimal solutions to problem~\eqref{eq: general optimization problem with precondition} and problem~\eqref{eq: opt-ss}.

Intuitively, the two equality constraints \eqref{eq: swing dynamics at equilibium for s} and \eqref{eq: power balance at equilibrium for s} jointly ensure that the optimal solution to problem \eqref{eq: general optimization problem with precondition} satisfies $\omega_i^*=0$ if one temporarily views $\tilde\theta$ and $\tilde\phi$ as the same. Then the additional quadratic term $\omega^TD\omega/2$ in the objective function does not alter the optimal solution. This observation indicates that the designed optimization problem \eqref{eq: general optimization problem with precondition} is equivalent to the optimal steady-state problem \eqref{eq: opt-ss}. The following lemma formalizes this equivalence.
\begin{lem}[Equivalence of optimization problems]\label{lem: equivalence of two optimization problem}
Suppose Assumption~\ref{ass: f} holds. A point $(\tilde\theta^*, \omega^*, s^*, \tilde\phi^*)$ is an optimal solution to \eqref{eq: general optimization problem with precondition} if and only if $(\tilde\theta^*, \omega^*, u^*)$ is an optimal solution to \eqref{eq: opt-ss}, with $u^*=f(s^*)$ and $\tilde\theta^* = \tilde\phi^*$.
\end{lem}

%With Lemmas~\ref{lem: unique global optimal} and \ref{lem: equivalence of two optimization problem}, we now claim that the closed-loop equilibrium achieves the optimal steady-state objectives by the following theorem. 
With Lemmas~\ref{lem: unique global optimal} and \ref{lem: equivalence of two optimization problem} in mind, we now claim that the closed-loop equilibrium attains the optimal solution to problem~\eqref{eq: general optimization problem with precondition}, as stated in the following theorem, thereby achieving the optimal steady-state objective. 
\begin{thm}[Closed-loop equilibrium]\label{thm: closed-loop equilibrium}
    Given $\omega^* = \nu^*$, a point $(\tilde\theta^*, \omega^*, s^*, \lambda^*, \tilde\phi^*)$ is an equilibrium of the closed-loop system \eqref{eq: swing dynamics in vector form}, \eqref{eq: u=f(s)}, and \eqref{eq: learnable controller} if and only if $(\tilde\theta^*, \omega^*, s^*, \nu^*, \lambda^*, \tilde\phi^*)$ is a primal-dual optimal solution to problem~\eqref{eq: general optimization problem with precondition}.
\end{thm}

With the characterization of the closed-loop equilibrium, we define the equilibrium set as
$$E := \big\{(\tilde\theta, \omega, s, \lambda, \tilde\phi)\ \big|\ \dot{\tilde\theta}, \dot{\omega}, \dot{s}, \dot{\lambda}, \dot{\tilde\phi} = 0\big\}\,,$$
and establish the stability of controller~\eqref{eq: u=f(s)} and \eqref{eq: learnable controller} through the following theorem.

\begin{thm}[Asymptotic stability]\label{thm: asymptotic stability}
    Suppose Assumption~\ref{ass: f} holds. Starting from any initial point in $\mathbb{R}^{3|\mathcal{V}|+2|\mathcal{E}|}$, the trajectory of \eqref{eq: swing dynamics in vector form}, \eqref{eq: u=f(s)}, and \eqref{eq: learnable controller} asymptotically converges to the equilibrium in $E$. 
\end{thm}

%To prove the asymptotic convergence, by following \cite{feijer2010stability}, we consider the candidate Lyapunov function
%\begin{equation}
    %V(z, s) \!= \!\frac{1}{2} (z\! - \!z^*)^T \Gamma (z \!- \!z^*)+ \frac{1}{2} (f(s) \!- \!f(s^*))^T (f(s)\! - \!f(s^*)) 
%\end{equation}
%with $z := (\tilde\theta, \omega, \lambda, \tilde\phi)$. $(z^*,s^*) = (\tilde\theta^*, \omega^*, s^*, \lambda^*, \tilde\phi^*)$ is the equilibium. 
% $\Gamma := \operatorname*{diag}(B^{-1}, M^{-1}, \Gamma^\lambda, \Gamma^{\tilde\phi})$ 
%\textcolor{red}{$\Gamma := \operatorname*{diag}(B, M, (\Gamma^\lambda)^{-1}, (\Gamma^{\tilde\phi})^{-1})$} is a diagonal positive definite constant matrix. The proof is given in three steps: 
%
%\begin{enumerate}
%    \item $V(z_s)\geq0$ and equality holds if and only if $z=z^*$ and $s=s^*$.

%    \item $\dot{V}(z_s)\leq 0$ and the trajectory is bounded.

%    \item With LaSalle's invariance principle, any trajectory starting from the largest invariant subset $E^*$ of $E$ will asymptotically converge to $E^*$.

%\end{enumerate}

%Up to this point, we conclude that the proposed learnable controller~\eqref{eq: u=f(s)} and \eqref{eq: learnable controller} asymptotically stabilize the system~\eqref{eq: swing dynamics in vector form} to an equilibrium, which is the optimal solution to problem~\eqref{eq: opt-ss}. Thus, the optimal steady-state objective is achieved.

\begin{rem}[Change of variables]
    With theorem~\ref{thm: closed-loop equilibrium} and \ref{thm: asymptotic stability}, we have shown that the nonlinear change of variables preserves both asymptotic stability and steady-state economic optimality, provided that the transformation is strictly monotone. More importantly, it is this nonlinear change of variables that embeds flexibility for learning within the primal-dual framework systematically. 
    Furthermore, this interpretation reveals the potential to extend the learning-augmented control design towards a broader range of optimization objectives, including, for instance, the operational constraints.
%address a broader range of optimization objectives, even considering operational constraints. As long as it can be properly modeled as an optimization problem, the controllers with flexibility for learning can be systematically designed.
\end{rem}

%% file: Neural_Network.tex
We have proposed a class of controllers that achieve asymptotic stability with economically optimal steady-state operation. Within this class, we can search for the best controller under transient metrics of interest by parameterizing the change of variables $u=f(s)$ as a monotone neural network and training it through reinforcement learning. In this section, we first design a new transient performance metric that provides theoretical guarantee for learning, followed by neural network training.
%Our overall goal is to design a controller that achieves the optimal steady-state objectives with improved transient performance. To this end, we have structured the controller to inherently achieve the optimal steady-state objectives and embed flexibility for learning. The next step is therefore to design appropriate metrics for transient performance of interest and guide the controller toward improved transient behavior through learning.
% This section first presents the construction of a strictly monotone neural network, followed by discretization of the closed-loop system and neural network training.

\subsection{Transient Performance Metrics}\label{sec: transient performance metrics}

For transient performance, we first focus on the convergence rate. Theorem~\ref{thm: asymptotic stability} only guarantees that the system converges to the equilibrium as $t \to \infty$, but it cannot guarantee the convergence rate.
Therefore, we introduce a metric to explicitly quantify the transient decay behavior
\begin{equation}\label{eq: convergence rate}
        R_{i,T}^\alpha :=\int _{0}^T e^{\alpha t}|\omega_i(t)|^2 dt\,,
    \end{equation}
where $\alpha >0$ denotes the exponential rate, reflecting the theoretical convergence speed of system frequency.
\begin{thm}[Potentially exponential stability]\label{thm: exponential convergence}
    Starting from any initial point in $\mathbb{R}^{3|\mathcal{V}|+2|\mathcal{E}|}$, suppose $T \to \infty$ and $R^\alpha_T := \sum_{i=1}^n R_{i,T}^\alpha$ converges during training. Then the learned controller will exponentially stabilize the frequency deviation $\omega$ with rate at least $\alpha/2$.
\end{thm}

Theorem~\ref{thm: exponential convergence} provides a theoretical guarantee that the controller drives the system to the equilibrium with an exponential response, under the assumption of ideal condition and learnability, i.e., $T \to \infty$ and $R_T^\alpha$ converges. Although such an ideal condition cannot be fully achieved in practice, the metric~\eqref{eq: convergence rate} still accelerates convergence and improves the controller’s response rate.
%From the learning perspective, the designed metric~\eqref{eq: convergence rate} plays a similar role in enhancing the controller’s ability to stabilize the system more rapidly. This further illustrates the motivation and significance of interpreting the controller from a primal-dual perspective.

Besides the convergence rate, we are also concerned with the maximum frequency deviation (frequency nadir)
\begin{equation}
        \|\omega_i\|_\infty := \max_{t\geq0}|\omega_i(t)|
 \end{equation}
and accumulated control effort over the transient period

\begin{equation}
        \bar C_{i,T}:=\frac{1}{T}\int_0^TF_i(u_i)dt\,.
\end{equation}

Finally, we define the following loss function to evaluate the transient performance of the controller:
\begin{equation}
    J := \sum_{i=1}^n(\rho_rR_{i,T}^\alpha + \rho_n\|\omega_i\|_\infty+\rho_c \bar C_{i,T})\,,
\end{equation}
where $\rho_r, \rho_n, \rho_c>0$ are tunable weights that balance the three transient metrics.

\subsection{Strictly Monotone Neural Network Training}\label{sec: NN training}
According to~\cite[Theorem~2]{Cui2023tps} and universal approximation property~\cite{lu2017expressive}, any strictly increasing function can be parameterized by a single-hidden-layer neural network with ReLU activation. To parameterize the strictly monotone function $u=f(s)$, we employ the stacked ReLU monotonic neural network in~\cite{jiang2022ojcsys}. By properly designing the relationship between the trainable parameters, the piece-wise linear network is guaranteed to be strictly monotone~\cite[Lemma 5]{Cui2023tps}.

To emulate the closed-loop dynamics, we adopt the recurrent neural network (RNN) structure. Specifically, we follow~\cite{jiang2022ojcsys} to discretize \eqref{eq: swing dynamics in vector form}, \eqref{eq: u=f(s)}, and \eqref{eq: learnable controller} with a time step $h$. 

In training procedure, one batch contains $\lfloor T/h \rfloor$ RNN cells. Each cell receives $(\tilde\theta^{\langle\tau\rangle},\omega^{\langle\tau\rangle},s^{\langle\tau\rangle},\lambda^{\langle\tau\rangle},\tilde\phi^{\langle\tau\rangle})$ as input and outputs $(\tilde\theta^{\langle\tau+1\rangle},\omega^{\langle\tau+1\rangle},s^{\langle\tau+1\rangle},\lambda^{\langle\tau+1\rangle},\tilde\phi^{\langle\tau+1\rangle})$ to the next cell with $\tau=0, 1, \dots, \lfloor T/h \rfloor$. Over $|\mathcal{B}|$ batches, we compute the loss 
\begin{equation}
    J(\Theta_i(d)):=\frac{1}{|\mathcal{B}|}\sum_{b=1}^{|\mathcal{B}|}\sum_{i=1}^n(\rho_r\hat{R}_{i,T}^\alpha+\rho_n\|\omega_i^{\langle T\rangle}\|_\infty+\rho_c\hat{\bar{C}}_{i,T})
\end{equation}
and update the trainable parameter set $\Theta_i(d)$ via gradient descent. One such update constitutes an epoch, and the training proceeds for $N$ epochs until $J(\Theta_i(d))$ converges.

%% file: Numerical_Results.tex
In this section, we conduct simulations on the IEEE 39-bus New England test system \cite{athay1979practical}, which comprises 39 buses interconnected by 46 transmission lines. Among them, 10 buses (indexed by $i \in \{30, \dots, 39 \}$) host generators, each with one generator per bus.
The system parameters, except $D$, are adopted from the Power System Toolbox~\cite{chow1992toolbox}, and the droop coefficients are set as $D_i = 150$ p.u. for $i \in \{30, \dots, 39\}$ and $D_i = 100$ p.u. for $i \in \{1, \dots, 29\}$. The generation cost functions are specified as $F_i(u_i) = \tfrac{1}{4}c_i u_i^4 + b_i$, $\forall i \in \{30, \dots, 39 \}$, where the coefficients $c_i$ and $b_i$ are randomly drawn from $(0,1)$. 
In what follows, we first present the training results of the monotone neural network used to parameterize the function $u_i = f_i(s_i)$. We then present the performance of the trained controller in comparison with a traditional linear primal-dual based controller.
%We apply a step change of $p_i=3 \text{ p.u.}$ in uncontrollable load at bus $i \in\{14, 22, 28\}$, respectively.
\subsection{Training of Monotone Neural Network}
To implement the training illustrated in Section~\ref{sec: NN training}, each trajectory sample is obtained by applying random step changes $p$ to the discretized system model. For each bus $i$, the disturbance $p_i$ is independently drawn from a uniform distribution over $(-5, 5)$ p.u.. The hyperparameters of neural network with 20 hidden neurons are set as follows: $|\mathcal{B}| = 64$, $N = 50$, $\alpha = 3$, $\rho_r = 0.1$, $\rho_n = 1$, and $\rho_c = 1$. The initial learning rate is $0.4$, which decays exponentially by a factor of $0.5$ every $3$ steps. The time series length is $6000$, and the optimizer used is Adam.
\begin{figure}[htbp]
  \centering  
  \includegraphics[width=0.99\linewidth]{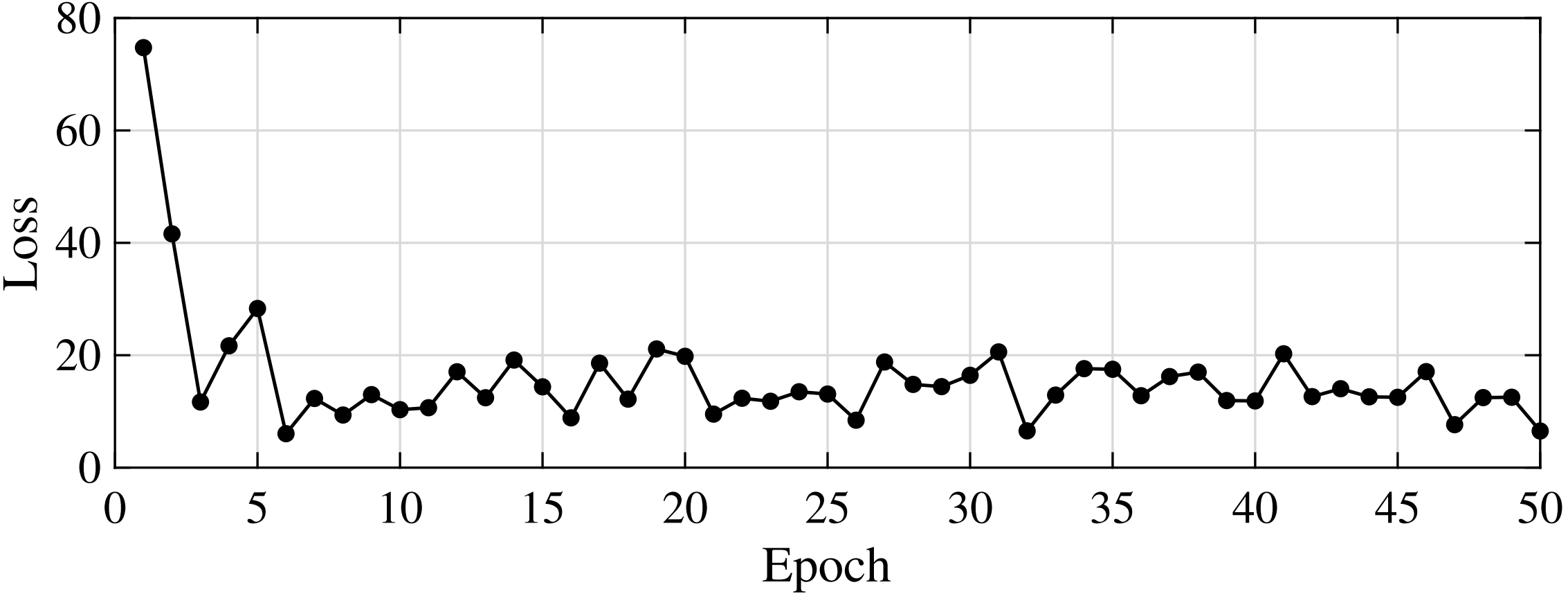} 
  \caption{Loss over epochs in training}  
  \label{fig: loss}  
\end{figure}

The training is conducted with TensorFlow 2.20.0 on an A800 GPU, which takes approximately 16,320 seconds. The evolution of the loss during training is shown in Fig.~\ref{fig: loss} that the loss $J$ is reduced by approximately $86.7\%$.

\subsection{Performance of the Learned Controller}
We now evaluate the performance of the trained controller and compare it with the traditional primal-dual based linear controller without learning. Disturbance $p_i = 3$ p.u. is applied on bus $i \in \{14, 22, 28, 36, 38\}$ at $t = 0$. Fig.~\ref{fig: frequency with controllers} shows the frequency responses and convergence on generator buses $i \in \{30, \dots, 39\}$, under both controller. Fig.~\ref{fig: cost with controllers} compares the control effort of both controllers during transient.
According to the results, the learned controller achieves significant improvements in convergence speed, with smaller frequency nadir and reduced control effort during transients. This advantage is further confirmed by the numerical comparison provided in Table~\ref{tab: transient performance comparison}.

\begin{figure}[htbp]
  \centering  
  \includegraphics[width=0.99\linewidth]{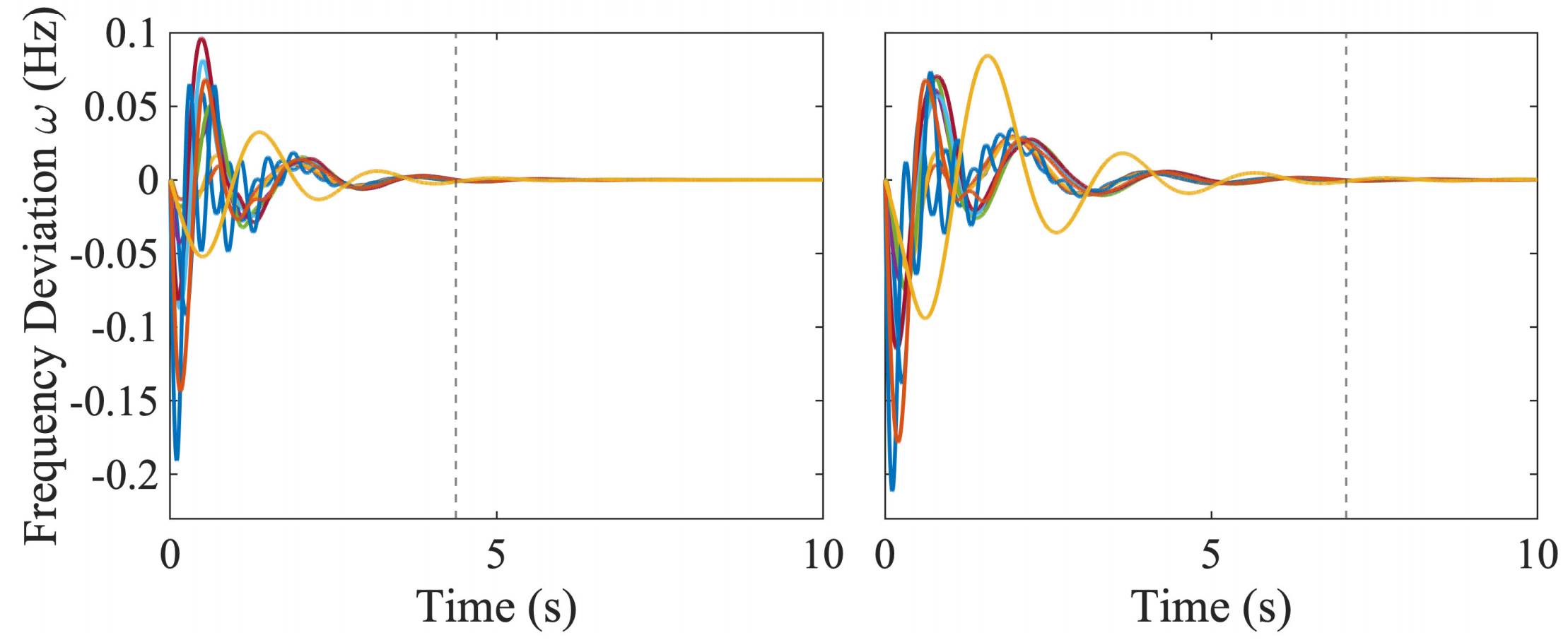} 
  \caption{Frequency deviation of generator buses under a disturbance $p_i = 3$ p.u. on buses $\{14,22,28,36,38\}$: learned controller (left); traditional primal-dual based controller (right).}
  \label{fig: frequency with controllers}
\end{figure}

\begin{figure}[htbp]
  \centering
    \centering  
  \includegraphics[width=0.99\linewidth]{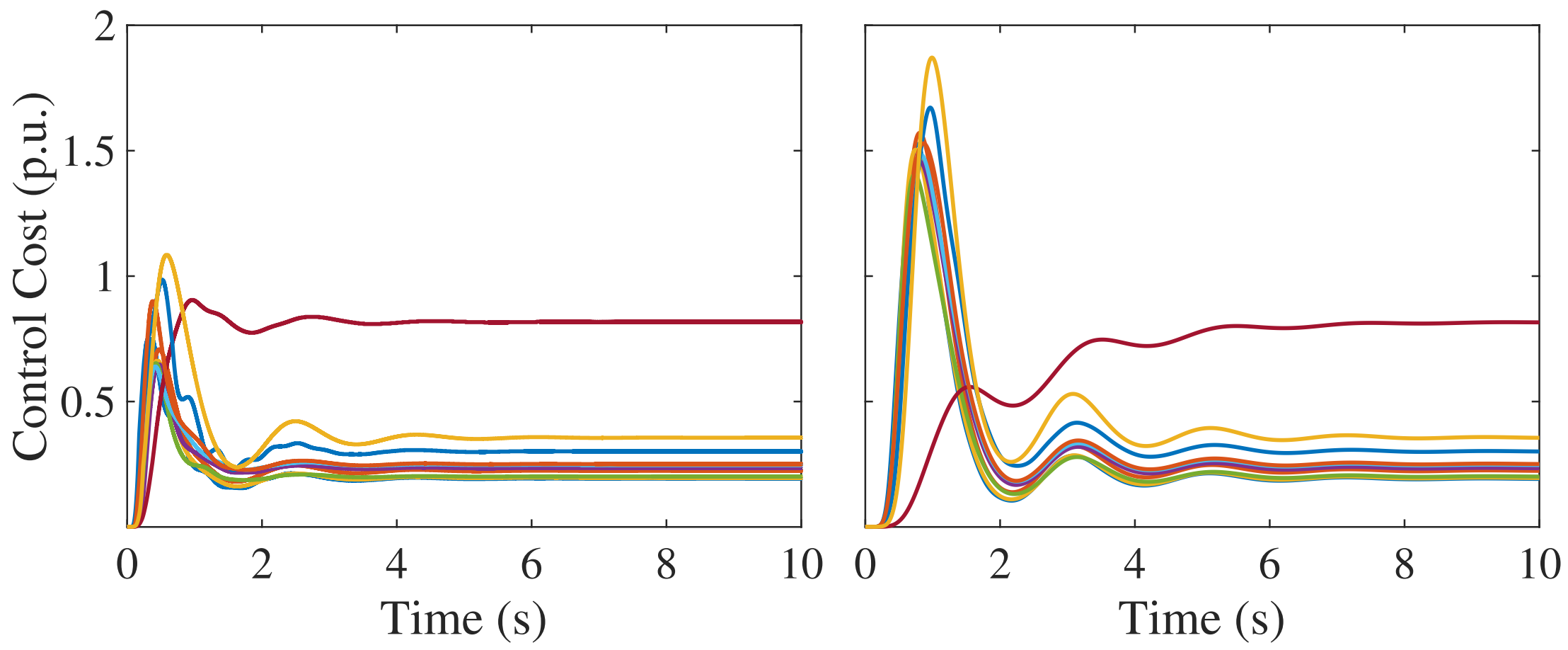} 
  \caption{Cost of generators under a disturbance $p_i = 3$ p.u. on buses $\{14,22,28,36,38\}$: learned controller (left); traditional primal-dual based controller (right).}
  \label{fig: cost with controllers}
\end{figure}

\begin{table}
    \centering
    \renewcommand{\arraystretch}{1.8} % 调整表格行高
    \begin{tabular}{>{\centering\arraybackslash}p{0.3\linewidth}|>{\centering\arraybackslash}p{0.16\linewidth}|>{\centering\arraybackslash}p{0.16\linewidth}|>{\centering\arraybackslash}p{0.16\linewidth}}
    \hline\hline  % 顶部双线
 & rate (s)& nadir&cost\\\hline
         learned controller&  \textbf{4.3825}&            \textbf{0.1680}& \textbf{123.3}\\\hline
         linear controller&  7.0695&  0.1796&133.8\\\hline \hline
    \end{tabular}
    \caption{Transient performance comparison of the controllers.}
    \label{tab: transient performance comparison}
\end{table}

To demonstrate that the learned controller achieves economic optimality at steady state, we plot the marginal cost of each generator in Fig.~\ref{fig: marginal cost}. In unconstrained optimization, identical marginal costs across generators indicate steady-state optimality~\cite[Chapter 5.5]{Boyd2004convex}. As shown in Fig.~\ref{fig: marginal cost}, the baseline primal–dual controller (right) inherently guarantees economic optimality. Our proposed controller (left) preserves economic optimality while incorporating a nonlinear learnable component, as evidenced by the identical marginal costs.

\begin{figure}[htbp]
  \centering  
  \includegraphics[width=0.99\linewidth]{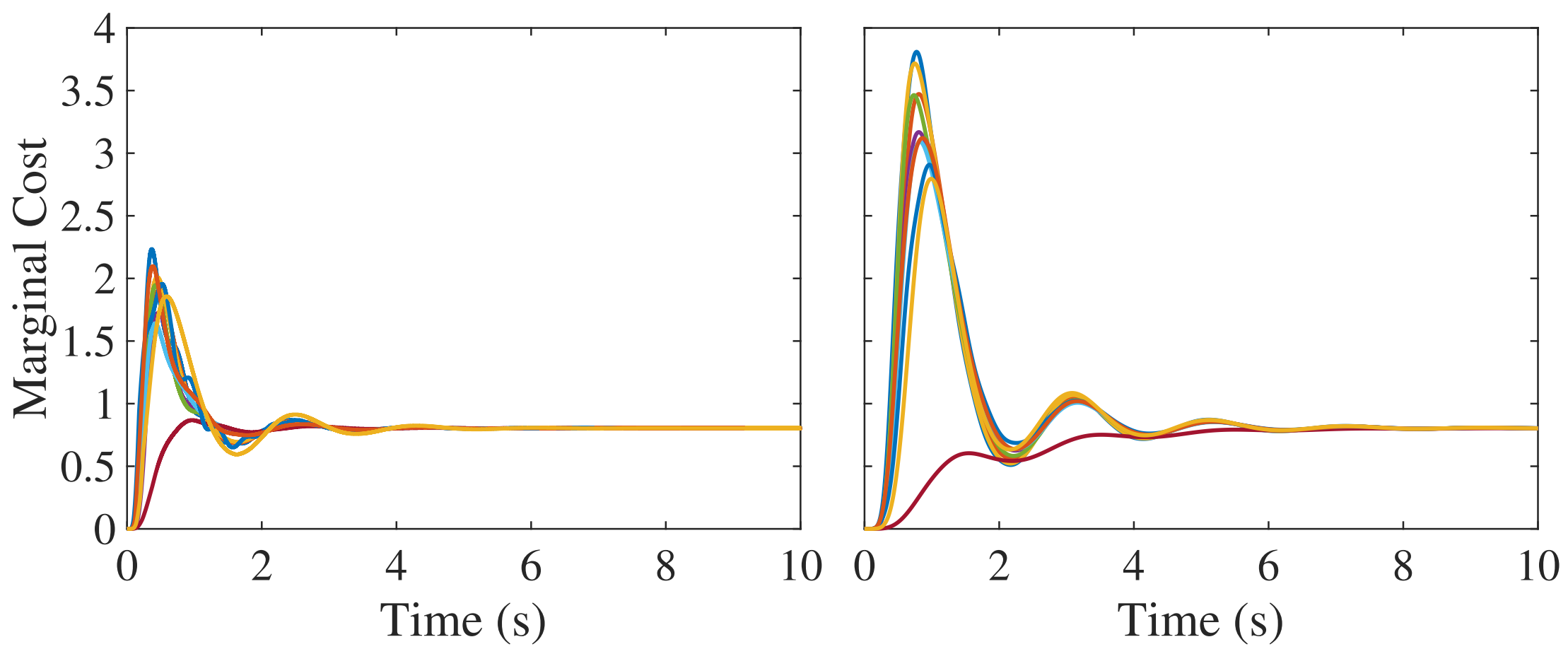} 
  \caption{Identical marginal cost among generators}  
  \label{fig: marginal cost}  
\end{figure}

%% file: Conclusion.tex
In this work, we propose a controller that simultaneously achieves secondary frequency regulation, steady-state economic optimality, and improved transient performance. By extending the primal–dual framework, we reinterpret the closed-loop system from an optimization perspective, which offers a unified framework for integrating learning flexibility into controller design. This systematic approach also enables addressing broader objectives, including operational constraints. 
To guarantee effectiveness of learning, we further propose a metric for convergence rate that provides theoretical assurance of potentially exponential convergence of the system under the learned controller. Experiments on the IEEE 39-bus system show that the controller achieves faster convergence, smaller frequency nadir and less accumulated control cost during transients. These results validate that our approach effectively integrates stability, optimality, and learnability in a unified control design.

%% file: Appendix.tex
\subsection{Proof of Lemma~\ref{lem: unique global optimal}}
We first construct a new optimization problem
\begin{subequations}\label{eq: general optimization problem}
\begin{align}
    \min_{\tilde\theta, \omega, u, \tilde\phi} &\quad F(u)+\frac{1}{2}\omega^TD\omega \label{eq: general objective function}\\
    {\operatorname*{s.t.}} &\quad u-p-D\omega-CB\tilde\theta=0 \label{eq: swing dynamics at equilibrium}\\
    &\quad u-p-CB\tilde\phi=0 \label{eq: power balance at equilibrium}
\end{align}
\end{subequations}
and give the following lemma.
\begin{lem}\label{lem: uniqueness of optimality of general optimization problem}
    If $(\tilde\theta^*, \omega^*, u^*, \tilde\phi^*)$ is a global optimal solution to \eqref{eq: general optimization problem}, then
    \begin{enumerate}
        \item $\omega^*=0$;\label{item: omega^*=0}

        \item $(\tilde\theta^*, \omega^*, u^*, \tilde\phi^*)$ is unique.\label{item: uniqueness of x_u^*}
    \end{enumerate}
\end{lem}
\begin{proof}
    To prove \ref{item: omega^*=0}), we suppose that $(\tilde\theta^*, \omega^*, u^*, \tilde\phi^*)$ is optimal but $\omega^* \neq 0$. Since $(\tilde\theta^*, \omega^*, u^*, \tilde\phi^*)$ is feasible, then we have $u^*-p=CB\tilde\phi^*$.
Consider $(\tilde\theta^\sharp, \omega^\sharp, u^*, \tilde\phi^*)$ with $\tilde\theta^\sharp=\tilde\phi^*$ and $\omega^\sharp=0$. It is feasible to \eqref{eq: general optimization problem} and satisfies 
$$F(u^*)+\frac{1}{2}({\omega^{\sharp}})^TD\omega^\sharp < F(u^*)+\frac{1}{2}({\omega^*})^TD\omega^*\,.$$
This is contradictory to optimality of $(\tilde\theta^*, \omega^*, u^*, \tilde\phi^*)$. So we have $\omega^*=0$. 

To prove \ref{item: uniqueness of x_u^*}), we first prove that $u^*$ is unique. Suppose there exists another optimal solution $(\tilde\theta^\star, \omega^*, u^\star, \tilde\phi^\star)$ with $u^\star \neq u^*$ and $F(u^*) = F(u^\star)$. 
% Suppose $u^\star \neq u^*$ such that $(\tilde\theta^*, \omega^*, u^\star, \tilde\phi^*)$ is also optimal. 
Then the convex combination of $(\tilde\theta^*, \omega^*, u^*, \tilde\phi^*)$ and $(\tilde\theta^\star, \omega^*, u^\star, \tilde\phi^\star)$ 
% $(\tilde\theta^*, \omega^*, u^*, \tilde\phi^*)$ and $(\tilde\theta^*, \omega^*, u^\star, \tilde\phi^*)$ 
is feasible, since the constraints are affine. With the strict convexity of $F(u)$, we have $F(\tfrac{1}{2}u^*+\tfrac{1}{2}u^\star )+\tfrac{1}{2}({\omega^*})^TD\omega^* <\tfrac{1}{2}F(u^*)+\tfrac{1}{2}F(u^\star)+ \tfrac{1}{2}({\omega^*})^TD\omega^*=F(u^*) + \tfrac{1}{2}({\omega^*})^TD\omega^*$, 
% $F(u)$ is strictly convex and $\omega^*=0$, then $F(u^*)+\tfrac{1}{2}({\omega^*})^TD\omega^* > F(\tfrac{1}{2}u^*+\tfrac{1}{2}u^\star )+\tfrac{1}{2}({\omega^*})^TD\omega^* $, 
which is contradictory to the optimality of $(\tilde\theta^*, \omega^*, u^*, \tilde\phi^*)$. Hence $u^*$ is unique. 

Since $\omega^*=0$, \eqref{eq: swing dynamics at equilibrium} at optimal can be rewritten as $u^*-p-CB\tilde\theta^*=0$. With $\tilde\theta=C^T\theta$, we obtain $CB\tilde\theta^*=CBC^T\theta^*=L\theta^*$ where $L$ is the Laplacian matrix of the connected undirected graph corresponding to $(\mathcal{V}, \mathcal{E})$. By the feasibility condition, this implies $L\theta^* = u^* - p$. According to graph theory, we have $\operatorname*{ker}(L)=\operatorname*{span}\{ 1_n \}$. Therefore, $\theta^*$ is determined uniquely up to a constant and consequently $\tilde\theta^*=C^T\theta^*$ is unique. By the same argument, $\tilde\phi^*$ is unique. This ends the proof of \ref{item: uniqueness of x_u^*}).
\end{proof}

Now we consider the following lemma to establish the equivalence of \eqref{eq: general optimization problem with precondition} and \eqref{eq: general optimization problem}.
\begin{lem}\label{lem: equivalence by change a variable}
    Suppose Assumption~\ref{ass: f} holds and let $u^*=f(s^*)$, then $(\tilde\theta^*, \omega^*, s^*, \tilde\phi^*)$ is globally optimal to problem~\eqref{eq: general optimization problem with precondition} if and only if $(\tilde\theta^*, \omega^*, u^*, \tilde\phi^*)$ is globally optimal to problem~\eqref{eq: general optimization problem}.
    %\begin{enumerate}
        %\item $(\tilde\theta^*, \omega^*, s^*, \tilde\phi^*)$ is a global optimal solution to problem~\eqref{eq: general optimization problem with precondition};\label{item: 1 in 2}

        %\item $(\tilde\theta^*, \omega^*, u^*, \tilde\phi^*)$ is a global optimal solution to problem~\eqref{eq: general optimization problem};\label{item: 2 in 2}

        %\item there exist unique $(\nu^*, \lambda^*)$ that $(\tilde\theta^*, \omega^*, s^*, \tilde\phi^*, \nu^*, \lambda^*)$ satisfies the Karush–Kuhn–Tucker (KKT) conditions of optimization problem~\eqref{eq: general optimization problem with precondition}.\label{item: 3 in 3}
    %\end{enumerate}
\end{lem}

\begin{proof}
    $\Longrightarrow$: $(\tilde\theta^*, \omega^*, s^*, \tilde\phi^*)$ is a global optimal solution to problem \eqref{eq: general optimization problem with precondition} with $u^*=f(s^*)$. Since $f$ is a strictly increasing bijection (by Assumption~\ref{ass: f}), the mapping $u = f(s)$ is invertible.

Let $(\hat{\tilde\theta},\hat{\omega}, \hat{u}, \hat{\tilde\phi})$ be any feasible point of problem \eqref{eq: general optimization problem}. With $\hat{u}=f(\hat{s})$, we have 
$$
F(f(s^*)) + \tfrac{1}{2}(\omega^*)^T D \omega^* \leq F(f(\hat{s})) + \tfrac{1}{2}\hat{\omega}^T D \hat{\omega}\,.
$$
This is equivalent to
$$
F(u^*) + \tfrac{1}{2}(\omega^*)^T D \omega^* \leq F(\hat{u}) + \tfrac{1}{2} \hat{\omega}^T D \hat{\omega}\,,
$$
with the change of variables. Since $(\tilde\theta^*, \omega^*, u^*, \tilde\phi^*)$ is also feasible for problem \eqref{eq: general optimization problem}, we conclude that it is the unique global optimal solution to \eqref{eq: general optimization problem}, where the uniqueness follows from Lemma~\ref{lem: uniqueness of optimality of general optimization problem}.

%%%%%%%%%%%%%%%%%%%%%%%%%%%%%%%%%%%%%%%%%%%%%%%%%%%%%%%%%%%%%%%%
$\Longleftarrow$: $(\tilde\theta^*, \omega^*, u^*, \tilde\phi^*)$ is the unique global optimal solution to problem \eqref{eq: general optimization problem}. With $s^* = f^{-1}(u^*)$, we have
% Then
$$F(f(s^*)) + \tfrac{1}{2}(\omega^*)^T D \omega^* < F(f(\hat{s})) + \tfrac{1}{2} \hat{\omega}^T D \hat{\omega}$$  
for any feasible point $(\hat{\tilde\theta},\hat{\omega}, \hat{s}, \hat{\tilde\phi})$ of \eqref{eq: general optimization problem with precondition} with $(\hat{\tilde\theta},\hat{\omega}, f(\hat{s}), \hat{\tilde\phi})\neq (\tilde\theta^*, \omega^*, u^*, \tilde\phi^*)$. Thus, $(\tilde\theta^*, \omega^*, s^*, \tilde\phi^*)$ is a unique global optimal solution to \eqref{eq: general optimization problem with precondition}.
\end{proof}

Based on Lemmas~\ref{lem: uniqueness of optimality of general optimization problem} and \ref{lem: equivalence by change a variable}, the global optimal solution $(\tilde\theta^*, \omega^*, s^*, \tilde\phi^*)$ to \eqref{eq: general optimization problem with precondition} is unique. We next establish the uniqueness of the associated Lagrange multiplier.

We first verify the Linear Independence Constraint Qualification (LICQ). Let $g_1(\tilde\theta,\omega,s):=f(s)-p-D\omega-CB\tilde\theta$ and $g_2(\tilde\phi,s):=f(s)-p-CB\tilde\phi$. The Jacobian $J_g$ of the equality constraints $(g_1,g_2)$ is written as
$$
J_g=
\begin{pmatrix}
-CB & -D & \nabla_sf & 0\\
0 & 0 & \nabla_sf & -CB
\end{pmatrix}\,.
$$
With the positive definiteness of $D$ and Assumption~\ref{ass: f}, we have $\operatorname*{rank}(J_g)=2n$. Thus, the equality constraints in \eqref{eq: general optimization problem with precondition} are linearly independent and LICQ holds.

Since $(\tilde\theta^*, \omega^*, s^*, \tilde\phi^*)$ is a global minimizer and the constraints are differentiable and satisfy LICQ at $(\tilde\theta^*, \omega^*, s^*, \tilde\phi^*)$, the classical Karush–Kuhn–Tucker (KKT) necessary conditions apply~\cite[Chapter 5.5.3]{Boyd2004convex}: there exists at least one pair of multipliers $(\nu^*,\lambda^*)$ such that $(\tilde\theta^*, \omega^*, s^*, \tilde\phi^*,\nu^*,\lambda^*)$ satisfy the KKT conditions (primal feasibility and stationarity). 

Uniqueness of the multiplier pair follows directly from LICQ. Suppose $(\nu^1,\lambda^1)$ and $(\nu^2,\lambda^2)$ are two multiplier pairs that both satisfy the stationarity conditions at the same primal point $(\tilde\theta^*, \omega^*, s^*, \tilde\phi^*)$. Subtracting the two stationarity relations yields a vector $(\delta_\nu,\delta_\lambda):=(\nu^1-\nu^2,\;\lambda^1-\lambda^2)$ such that 
$$J_g^T(\delta_\nu,\delta_\lambda)^T = 0\,.$$
But by LICQ, the only solution is $(\delta_\nu,\delta_\lambda)=(0,0)$. Hence $\nu^1=\nu^2$ and $\lambda^1=\lambda^2$. The multiplier pair is unique.
% This ends the proof of Lemma~\ref{lem: unique global optimal}

With convexity of \eqref{eq: general optimization problem}, the linear independence of its equality constraints, and uniqueness of global minimizer, $(\tilde\theta^*, \omega^*, u^*, \tilde\phi^*)$ together with $\nu^*$ and $\lambda^*$ constructed above is the unique solution to the KKT conditions of problem~\eqref{eq: general optimization problem}. By similar reasoning as in Lemma~\ref{lem: equivalence by change a variable},  the invertible monotone bijection $u = f(s)$ ensures that $(\tilde\theta^*, \omega^*, s^*, \tilde\phi^*,\nu^*,\lambda^*)$ is also the only solution satisfying the KKT conditions of \eqref{eq: general optimization problem with precondition}. This ends the proof of Lemma~\ref{lem: unique global optimal}.
%%%%%%%%%%%%%%%%%%%%%%%%%%%%%%%%%%%%%%%%%%%%%%%%%%

\subsection{Proof of Lemma~\ref{lem: equivalence of two optimization problem}}
Building upon Lemma~\ref{lem: equivalence by change a variable}, it is sufficient to establish the equivalence between problem \eqref{eq: general optimization problem} and problem \eqref{eq: opt-ss} at optimality. According to Lemma~\ref{lem: uniqueness of optimality of general optimization problem}, \eqref{eq: swing dynamics at equilibrium} can be written as 
$u^*-p-CB\tilde\theta^*=0.$
Together with \eqref{eq: power balance at equilibrium}, we obtain $CBC^T\theta^*=CBC^T\phi^*$ with $\tilde\theta^*=C^T\theta^*$ and $\tilde\phi^*=C^T\phi^*$. Following the same way in proving the uniqueness of $(\tilde\theta^*, \omega^*, u^*, \tilde\phi^*)$ with property of Laplacian matrix $L$, we have $C^T (\theta^* - \phi^*) = \tilde\theta^* - \tilde\phi^* = 0$.
%Using the property of $L = CBC^T$ with $\operatorname*{ker}(L)=\operatorname*{span}\{ 1_n \}$, it follows that
%CBC^T (\theta^* - \phi^*) = 0 \ \implies \  C^T (\theta^* - \phi^*) = \tilde\theta^* - \tilde\phi^* = 0.
% By Lemma~\ref{lem: equivalence by change a variable}, the same result holds for \eqref{eq: general optimization problem with precondition}. 
Therefore, at the optimal point, the equality constraints of problem~\eqref{eq: general optimization problem} can be written as 
$$u^*-p-D\omega^*-CB\tilde\theta^*=0, \quad u^*-p-CB\tilde\theta^*=0\,,$$
which has the same form as \eqref{eq: opt-ss} at optimal point. Moreover, since $\omega^* = 0$ at the optimum of both  \eqref{eq: opt-ss} and \eqref{eq: general optimization problem}, the quadratic term $\tfrac{1}{2}\omega^TD\omega$ does not affect the objective function. Hence, the two optimization problems \eqref{eq: opt-ss} and \eqref{eq: general optimization problem} share the same optimal solution for $\tilde \theta^*$, $\omega^*$ and $u^*$ with $\tilde \phi^* = \tilde \theta^*$. By Lemma~\ref{lem: equivalence by change a variable}, equivalence between \eqref{eq: opt-ss} and \eqref{eq: general optimization problem with precondition} is thereby established.

%%%%%%%%%%%%%%%%%%%%%%
\subsection{Proof of Theorem~\ref{thm: closed-loop equilibrium}}
The proof is based on the equivalence between the primal-dual dynamics~\eqref{eq: closed-loop system in primal-dual} and the closed-loop system~\eqref{eq: swing dynamics in vector form}, \eqref{eq: u=f(s)}, and \eqref{eq: learnable controller}, as well as the interchangable $\omega^* = \nu^*$.

$\Longrightarrow$: Suppose $(\tilde\theta^*, \omega^*, s^*, \lambda^*, \tilde\phi^*)$ is a closed-loop equilibrium. Then, by definition of equilibrium, derivatives of all variables with respect to time are all $0$. This implies stationarity of \eqref{eq: general optimization problem with precondition}. With $\dot\omega=0$ and $\dot \lambda=0$, we have $f(s^*)-p-D\omega^*-CB\tilde\theta^*=0$ and $f(s^*)-p-CB\tilde\phi^*=0$.
Thus, primal feasibility holds and $(\tilde\theta^*, \omega^*, s^*, \nu^*, \lambda^*, \tilde\phi^*)$ is a primal-dual optimal solution to \eqref{eq: general optimization problem with precondition}.

$\Longleftarrow$: suppose $(\tilde\theta^*, \omega^*, s^*, \nu^*, \lambda^*, \tilde\phi^*)$ is a primal-dual optimal solution to \eqref{eq: general optimization problem with precondition}. Then, by definition, the KKT conditions hold. The conditions imply that all right-hand sides of the closed-loop differential equations are zero. Therefore, $(\theta^*, \omega^*, s^*, \lambda^*, \tilde\phi^*)$ is an equilibrium of the closed-loop system.

%%%%%%%%%%%%%%%%%%%%%%%%%%%%%%%%%%%%%%%%%%%%%%%%%%%%%%%%%%%%%%%%%%%%%%%%%

\subsection{Proof of Theorem~\ref{thm: asymptotic stability}}

By following~\cite{jiang2022ojcsys}, we consider the following candidate Lyapunov function
\begin{equation}
\begin{aligned}
    V(z,s):=&\ \tfrac{1}{2} (z- z^*)^T \Gamma (z - z^*)\\
    &+ U(s)-U(s^*)-[\nabla U(s^*)]^T(s-s^*) \,,
\end{aligned}
\end{equation}
where $U(s):=\sum_{i=1}^n\int_0^{s_i}f_i(\xi)\mathrm{d}\xi$ and $z := ({\tilde \theta}, \omega, \lambda, \tilde \phi)$. $\Gamma:=\operatorname{diag}(B, M, (\Gamma^\lambda)^{-1}, (\Gamma^{\tilde\phi})^{-1})$ is positive definite. Point $(z^*, s^*)$ is the equilibrium of the closed-loop system.

%\textit{Step 1 ($V\geq0$ and $V=0$ iff at equilibrium):}
We first show that $V(z,s) \ge 0$, $\forall (z,s)$, and $V(z,s)=0$ if and only if the system is at the equilibrium.

With Assumption~\ref{ass: f}, $U(s)$ is a strictly convex function~\cite[Lemma 3]{jiang2022ojcsys}. Thus we have 
$$ U(s)-U(s^*)-[\nabla U(s^*)]^T(s-s^*)>0, \quad\forall s \neq s^*\,.$$
Together with $\Gamma \succ 0$, we have $V(z,s)\geq0$ and equality holds if and only if $(z,s)=(z^*, s^*)$.

Then we prove that $\dot V(z,s) \leq 0$ along the trajectories of the closed-loop system, and show that $\dot V(z,s)\equiv 0$ implies the system stays at the equilibrium, i.e., $\dot z \equiv 0$ and $\dot s \equiv 0$. 

Differentiating $V(z,s)$ over time, we obtain
$$\begin{aligned}
    \dot V(z,s) =& (z-z^*)^T \Gamma \dot z + (\nabla U(s)-\nabla U(s^*))^T  \dot s\\
     {=} &(z-z^*)^T \Gamma \dot z + (\nabla U(s)-\nabla U(s^*))^T  \dot s\\
     &-\omega^T(f(s^*)-p-CB\tilde\theta^*)\\
     &-\lambda^T(f(s^*)-p-CB\tilde\phi^*)\\
     {=} &-\omega^TD\omega-\omega^{*T}(f(s)-p-D\omega-CB\tilde\theta)\\
     &-\lambda^{*T}(f(s)-p-CB\tilde\phi)\\
     &-(f(s)-f(s^*))^T\nabla F(f(s))\\
     %\overset{\textcircled{3}\label{eq: step3 in Vdot}}{=} &-\omega^TD\omega-\lambda^{*T}(f(s)-p-CB\tilde\phi)\\
     %&-(f(s)-f(s^*))^T\nabla F(f(s))\\
     {=} &-\omega^TD\omega-(f(s)-f(s^*))^T\nabla F(f(s))\\
     &-\lambda^{*T}(f(s)-f(s^*)+CB\tilde\phi^*-CB\tilde\phi)\\
     {=} &-\omega^TD\omega-(f(s)-f(s^*))^T\nabla F(f(s)) \\
     &+(\nabla F(f(s^*)))^T(f(s)-f(s^*)) \\
      {=}& -\omega^TD\omega\\
      &-[\nabla F(f(s))-\nabla F(f(s^*))]^T(f(s)-f(s^*))\,.
      \end{aligned}$$
Here the second equality comes from $\dot{\omega}=\dot{\lambda}=0$ at the equilibrium and $\omega^*=0$. Then we have the third equality by removing the same terms. With $\omega^*=0$ and $f(s^*)-p-CB\tilde\phi^*=0$, we obtain the fourth equality. According to $\dot{\tilde\phi}=0$ and $\dot s=0$ at the equilibrium, we have $C^T\lambda^*=0$ and $\nabla F(f(s^*))=-\lambda^*$ with $\omega^*=0$, which yields the fifth and sixth equalities. 
% Then \textcircled{4} and \textcircled{5} are derived.

Now we analyze the properties of $\dot V(z,s)$. Since $D \succ 0$ and $\nabla F(\cdot)$ is strictly increasing, we obtain $\dot V(z,s) \leq 0$. Moreover, $\dot V(z,s) \equiv 0$ directly enforces $\omega \equiv 0$ and $f(s) \equiv f(s^*)$, hence $s \equiv s^*$ by the strict monotonicity of $f(\cdot)$. This further implies $\dot{\tilde \theta} \equiv C^T \omega \equiv 0$. Combining $\dot s = -[\nabla F(f(s)) + \omega + \lambda]$, we have $\lambda \equiv - \nabla F(f(s^*))$, thus $\dot \lambda \equiv 0$. Together with $s \equiv s^*$, $\dot \lambda = \Gamma^\lambda[f(s^*) - p - CB \tilde \phi] \equiv 0$ forces $C B \tilde \phi$ to be constant. Differentiating this constant term gives $C B \dot{\tilde \phi}=(C B\Gamma^{\tilde \phi}B C^T) \lambda \equiv 0$. As $C B\Gamma^{\tilde \phi}B C^T$ is the weighted Laplacian of a connected graph, this implies $\lambda \in \operatorname*{span}\{1_n \}$ and therefore $\dot{\tilde \phi}=BC^T\lambda \equiv 0$. 
% Since $D \succ 0$, and both $\nabla F(f(s))$ and $f(s)$ are strictly increasing functions passing through the origin, then we have $\dot V(z,s) \leq 0$.
This means the largest invariant set contained in the set of points where $\dot V$ vanishes is actually the equilibrium set $E$. By the LaSalle invariance principle~\cite{khalil2002nonlinear}, every trajectory of the closed-loop system~\eqref{eq: swing dynamics in vector form}, \eqref{eq: u=f(s)}, and \eqref{eq: learnable controller} converges to $E$, which contains only one equilibrium point characterized by Theorem~\ref{thm: closed-loop equilibrium}. This ends the proof for asymptotic stability.
%%%%%%%%%%%%%%%%%%%%%%%%%%%%%%%%%%%%%%%%%%%%%
\subsection{Proof of Theorem~\ref{thm: exponential convergence}}
Convergence of training implies $\lim_{T\to\infty} \sum_{i=1}^{n} R_{i,T}^\alpha = \sum_{i=1}^{n} \int_0^\infty e^{\alpha t} |\omega_i(t)|^2 dt < \infty$.
This means that each term $e^{\alpha t/2} \|\omega_i(t)\|$ is square-integrable over $[0,\infty)$. Because the system dynamics are continuous and $f(s)$ is Lipschitz, Barbalat's lemma~\cite{barbalat1959systemes} implies that $\lim_{t \to \infty} e^{\alpha t/2} |\omega_i(t)| = 0 $.
Since the function $e^{\alpha t/2} |\omega_i(t)|$ is continuous on $[0, \infty)$ and converges to a finite limit, it must be bounded by 
% The continuous function $e^{\alpha t/2} |\omega_i(t)|$ on $[0, \infty)$ that converges to a finite limit, then it is bounded with 
a constant $C_i$ \cite{rudin1976principles}, i.e., $|\omega_i(t)| \le C_i e^{-\alpha t/2}, \forall t \ge 0$.
This shows that $\omega_i(t)$ decays exponentially with rate at least $\alpha/2$.